\documentclass{ifacconf}
\usepackage{amssymb}
\usepackage{graphicx}      
\usepackage{natbib}        
\usepackage{amsmath}       
\usepackage{bm}
\usepackage{subcaption}

\newtheorem{remark}{Remark}

\newcommand{\reals}{\mathbb{R}}

\newcommand{\X}{\mathcal{X}}
\newcommand{\F}{\mathcal{F}}

\newcommand{\R}{\mathcal{R}}
\newcommand{\obs}{\mathcal{O}}

\newcommand{\D}{\partial}
\begin{document}
\setlength{\abovedisplayskip}{2pt}   
\setlength{\belowdisplayskip}{2pt}
\setlength{\abovedisplayshortskip}{2pt}
\setlength{\belowdisplayshortskip}{2pt}
\setlength{\parskip}{0.4em} 
\begin{frontmatter}

\title{Differentiable Optimization Layered Safety-Critical Control for Risk-Aware Navigation via Conformal Prediction} 

\thanks[footnoteinfo]{Jinyang Dong and Shizhen Wu are co-first authors. }

\thanks[footnoteinfo]{This work is supported in part by the National Natural Science Foundation of China under Grant U25A20473 and Grant 62233011; and in part by the Key Technologies R \& D Program of Tianjin under Grant 23YFZCSN00060.
E-mails: \{dongjinyang; szwu\}@mail.nankai.edu.cn, fangyc@nankai.edu.cn  \emph{(Corresponding author: Yongchun Fang)}.
 }

\author[First,Second]{Jinyang Dong}, 
\author[First,Third]{Shizhen Wu}, 
\author[First,Third]{Yongchun Fang}

\address[First]{Institute of Robotics and Automatic Information Systems, College of Artificial Intelligence, Nankai University, Tianjin 300071, China}
\address[Second]{Academy for Advanced Interdisciplinary Studies, Nankai University, Tianjin 300071, China}
\address[Third]{Institute of Intelligence Technology and Robotic Systems, Shenzhen Research Institute of Nankai University, Shenzhen 518083, China}

\begin{abstract}                
Risk-aware navigation in unknown environments is a fundamental challenge for autonomous vehicles operating in complex urban systems. To address this issue, this paper presents a differentiable optimization layered safety-critical control method based on conformal prediction.  
First, to handle uncertainties arising from sensor noise, the conformal prediction method is employed to generate risk-aware obstacle ellipsoids around an elliptical-shaped robot. 
Second, two nested differentiable optimization layers are introduced to build the control barrier functions for obstacle avoidance and feasibility guarantee, respectively. Then, a quadratic program based safety-critical control law is proposed to integrate the above control barrier function constraints as well as input constraints.
In the end, the effectiveness of the proposed framework is demonstrated through numerical simulations.
\end{abstract}


\end{frontmatter}

\section{Introduction}
Ensuring safety is a critical challenge in robotics and control systems, particularly for autonomous robots operating in urban environments. The control barrier function-quadratic program (CBF-QP) based safety-critical control framework (\cite{8796030}) has emerged as a powerful real-time tool to address these challenges.
Thanks to its low computational complexity and formal safety guarantees, CBF methods have seen widespread application in urban robotics and related fields (\cite{11180029,11124847, 11060906, shizhen2025Underactuated, lyu2024small}).

Recently, risk-aware navigation under measurement noise in unknown environments has received great attention.
Most existing safety controllers assume perfect knowledge of the environment, which is often unrealistic in practical scenarios where robots rely on noisy and uncertain sensor measurements to perceive their surroundings (\cite{long2025sensor}). 
These uncertainties can lead to inaccuracies in obstacle detection and localization. For the purpose of uncertainty quantification, conformal prediction (CP) (\cite{shafer2008tutorial}) has been introduced for risk-aware navigation in several aspects (\cite{yang2023safe, kwon2025conformalized, mei2024perceive}), where CP could provide valid prediction regions with a user-specified confidence level. More specifically,  CP is leveraged to dynamically calibrate prediction uncertainty and plan probabilistically safe trajectories for a legged robot in \cite{mei2024perceive}.
Similarly, a confidence bound is computed for the reachable set of manipulators using CP to provide a probabilistic safety guarantee (\cite{kwon2025conformalized}). Different from the above two works that apply CP to deal with the measurement noise of surrounding obstacles,  CP is employed in \cite{yang2023safe} and \cite{11007767} to quantify state estimation uncertainty and construct a probabilistic CBF. As far as the authors know,  how to further handle the measurement noise of surrounding obstacles using CP when constructing CBFs, has still not been considered.



Parallel to the aforementioned research progress, differentiable (parameteric) optimization (\cite{agrawal2019differentiable}) has been introduced in the area of CBF-based safety-critical control in two different aspects. On the one hand, differentiable optimization is utilized to handle the obstacle avoidance of shape-rich robots and obstacles. 
In \cite{kaypak2025safe} and \cite{11124847},  the scaling factor between two ellipsoids is chosen as a CBF for collision avoidance and solves the resulting constraints via QP.
Nevertheless, as described in \cite[Remark 1]{kaypak2025safe}, the presence of multiple constraints poses significant challenges to the feasibility of the resulting QP-based control laws. 
On the other hand, the feasibility guarantee problem has recently been solved by the (feasible space) volume-CBF method (\cite{parwana2025feasible, 11220207,wu2025polytope}), which could potentially be combined with differentiable optimization to accelerate computing efficiency. 
In detail, the volume of the largest inscribed ellipsoid/ball within the state-dependent polytope in the input space is defined as the volume-CBF in \cite{parwana2025feasible} and \cite{11220207}. 
By ensuring that the volume of this ellipsoid is strictly positive, the persistent feasibility of the QP under multiple constraints is achieved.
However, the gradient calculation of the volume of an ellipsoid typically relies on heuristic gradient calculations, which are time-consuming. 
In \cite{wu2025polytope}, the heuristic gradient calculation is avoided by utilizing the nonsmooth parameteric optimization and the nonsmooth CBF theory. How to further develop a more concise, differentiable optimization-based feasibility guarantee method remains to be studied.

Motivated by the above discussions, this paper represents both obstacles and the robot using ellipsoids, where obstacle ellipsoids are generated in real time based on the noisy sensor data. 
Besides, to account for the impact of perceptual noise on safety, CP is employed to quantify the uncertainties arising from sensor noise and prediction errors. 
In addition, a safety-critical controller based on multilayer differentiable optimization is proposed. 
First, inspired by \cite{kaypak2025safe}, the second-order conic program (SOCP) is used to compute the minimum scaling factor between ellipsoids for collision avoidance.
Following in line with the authors' prior work \cite{11220207}, the volume of the largest inscribed ellipsoid within the feasible space is obtained through a semidefinite program (SDP), serving as a persistent feasibility constraint.
Then, the control input is obtained by solving a QP that integrates these constraints. 
Among the above procedures, our main contributions are as follows:

\begin{itemize}
    \item A risk-aware minimum ellipsoid generation method under the feedback of online sensor data is given. In detail, CP is employed to calibrate the generated ellipsoids, so as to handle the uncertainty introduced by sensor noise and prediction errors. 
    \item  A SOCP-SDP-QP-based safety-critical control law is proposed, where
   two differentiable optimization layers (i.e., SOCP and  SDP) are introduced to build CBFs for collision detection and feasibility guarantee, respectively.  
    \item The efficiency of the proposed safety controller is validated in numerical simulation tests where a vehicle is driven to perform safe navigation in unknown environments.
\end{itemize}

The remaining parts of this paper are organized as follows.
Section \ref{sec:pf} gives the preliminaries and problem formulation.
Section \ref{sec:main} presents the main results of this paper, including the conformal prediction-based obstacle detection and the differentiable optimization layered safety controller.
Subsequently, several groups of simulation results are given in Section \ref{sec:sim} to validate the effectiveness of the proposed method.
Finally, Section \ref{sec:conclusion} summarizes this paper.

\emph{Notations:} For any integer $N$, $[N]=\{1, 2,\ldots,N\}$ denotes the set of integers
from $1$ to $N$.
A class function $C^{k}$ function is a function that has continuous derivatives up to order $k$. 
$\Gamma(\cdot)$ is the Gamma function.
The extended class -$\mathcal{K}_\infty$ function is a strictly increasing function $\alpha:\reals\to\reals$ with $\alpha(0)=0$. 
$L_fh(\bm x)=\frac{\D h}{\D \bm x}f(\bm x)$ and $L_gh(\bm x)=\frac{\D h}{\D \bm x}g(\bm x)$ are the Lie derivatives of $h$ along $f$ and $g$ at $\bm x$, respectively.

\section{PRELIMINARIES and problem Formulation}\label{sec:pf}
\vspace{-2mm}
Consider a nonlinear, control-affine system of the form
\begin{align}
  \label{equ:sys1}
  \dot{\bm x}&=f(\bm x)+g(\bm x)\bm u,
\end{align}
where $\bm x \in\mathcal{X}\subset\mathbb{R}^n$ is the state, $\bm u\in\mathcal{U}\subset\mathbb{R}^m$ is the control input.
$f:\mathcal{X} \to \mathbb{R}^n$ and $g:\mathcal{X}\to \mathbb{R}^{n\times m}$ are sufficiently smooth functions.
The set $\mathcal{U}$ is a bounded polytopic set, defined as $\mathcal{U}=\{\bm u\in \reals^m|\bm A_m\bm u\leq \bm b_m\}$, where $\bm A_m\in\reals^{2m\times m}$, $\bm b_m\in\reals^{2m}$ are both constant matrices.
Besides, the physical domain of the robot at state $\bm x\in \X$ is approximately modeled as an ellipsoid: 
\begin{align}
  &\R(\bm x)=\{\bm p\in\reals^{n_p}|\F_{\R}(\bm p, \bm \theta_\R(\bm x))\leq 1\},\\
  &\F_{\R}(\bm p, \bm \theta_\R(\bm x)) = (\bm p - \bm \mu_\R(\bm x))^\top \bm Q_\R(\bm x) (\bm p - \bm \mu_\R(\bm x)).
\end{align}

The parameter $\bm p$ is used to evaluate the function, with its dimension restricted to 2 or 3 (i.e., $n_p=2$ or $3$), representing to the 2D or 3D cases, respectively. The parameter $\bm \theta_{\mathcal{R}}$ is dependent on $\bm x$, in specific examples, $\bm \theta_{\mathcal{R}}$ is  usually determined by the position and orientation of the rigid shape $\R(\bm x)$. 

\vspace{-1mm}
\subsection{Preliminaries}
\vspace{-1mm}
The CBF based safety-critical control method in \cite{8796030} is first reviewed for the sake of completeness. 
The system (\ref{equ:sys1}) is called safe with respect to the set $\mathcal{S}$ if there exists a controller $\bm u$ rendering the set $\mathcal{S}$ forward invariant, i.e. for any initial state $\bm x(0)\in\mathcal{S}$, the resulting trajectory $\bm x(t)$ remains within $\mathcal{S}$ for all $t\geq 0$.

\emph{Definition 1 (CBF):} A continuously differentiable function $h:\mathcal{X}\to\mathbb{R}$ is a CBF for system (\ref{equ:sys1}) on the set $\mathcal{S}=\{\bm x\in\mathcal{X}|h(\bm x)\geq 0\}$ if there exists an extended class-$\mathcal{K}_\infty$ function $\phi$ such that for all $\bm x\in\mathcal{X}$,
\begin{align}
  \label{equ:cbf}
  \sup_{\bm u\in\mathcal{U}}[L_fh(\bm x)+L_gh(\bm x)\bm u]\geq -\phi(h(\bm x)).
\end{align}
For brevity, the function $\phi$ is often chosen as a linear function, i.e., $\phi(h(\bm x))=\kappa h(\bm x)$ with $\kappa>0$.

\emph{Lemma 1:} Given a CBF $h(\bm x)$ for system (\ref{equ:sys1}) on the set $\mathcal{S}$, any Lipschitz continuous controller $\bm u:\mathcal{X}\to\mathcal{U}$ that satisfies (\ref{equ:cbf}) renders 
the system (\ref{equ:sys1}) safe w.r.t the set $\mathcal{S}$. 

The following latest results about differentiable optimization-based collision detection are also introduced in this paper.

\emph{Definition 2 \cite[Scaling functions]{11124847}:} For a closed set $\mathcal{A}\subset\reals^{n_p}$ with non-empty interior, the function $\F_\mathcal{A}(\bm p, {\bm \theta}):\reals^{n_p}\times\reals^{n_\theta}\to\reals$ is called a scaling function with parameter $\bm \theta$ if $\F_\mathcal{A}$ is convex with respect to its first argument and satisfies
 $ \mathcal{A} = \{\bm p\in\reals^{n_p}|\F_\mathcal{A}(\bm p, \bm \theta)\leq 1\}$.

\emph{Lemma 2 \cite[Theorem 2]{11124847}:} Let $\F_\mathcal{O}$ and $\F_\mathcal{R}$ be scaling functions for the sets $\mathcal{O}$ and $\mathcal{R}$, respectively, with $\mathcal{O}\cap\mathcal{R}=\emptyset$. 
Assume that one of $\F_\mathcal{O}$ and $\F_\mathcal{R}$ has a positive definite Hessian with respect to $\bm p$. 
If the scaling functions are $C^{k+1}$ (with $k\geq 1$) with respect to $\bm p$ and $\bm \theta$, then
 the optimal value of the following parameteric optimization
\begin{subequations}
  \label{equ:scaling_factor_obs}
  \begin{align}
    \label{equ:scaling_factor_obj_obs}
    &\gamma^*(\bm x) = \min_{\bm p}~\F_{\obs}(\bm p, \bm \theta) \\
    \label{equ:scaling_factor_cons_obs}
    &\text{s.t.}~\F_{\R}(\bm p, \bm \theta_\R(\bm x)) \leq 1. 
  \end{align}
\end{subequations}
is $C^{k}$ in $\bm \theta$. 

In \cite{11124847}, it is also shown that  $\gamma^*(\bm x)>1$ if and only if the robot and obstacle have not intersected, that is, 
$ \obs\cap\R(\bm x)=\emptyset$.

\subsection{Problem Formulation}
\vspace{-2mm}
The following introduces the risk-aware navigation in an unknown environment. 
Assume that the obstacles are perceived by the sensors as a collection of point clouds $\{P_i\}_{i=1}^N$, where $P_i=\{\bm p_{i,j}\}_{j=1}^{n_i}$ corresponds to the point cloud associated with the $i$-th obstacle. 

Referring to \cite{kaypak2025safe},  the obstacles around the robot are represented by ellipsoids.
For each point cloud, the parameters of the minimum volume enclosing ellipsoid are obtained through a fitting function 
\begin{align}
\label{equ:fitting_function}
\hat{\bm \theta}_i = f_\text{Fit}(P_i),
\end{align}
where $\hat{\bm \theta}_i$ denotes the parameters of the fitted ellipsoid (e.g., center $\hat{\bm \mu}_i$ and shape matrix $\hat{\bm Q}_i$). Then, the $i$-th elliptical obstacle can be represented as
\begin{align}
&\hat \obs_i=\{\bm p\in\reals^{n_p}|\F_{\hat{\obs}_i}(\bm p, \hat{\bm \theta}_i)\leq 1\},\\
&\F_{\hat{\obs}_i}(\bm p, \hat{\bm \theta}_i) = (\bm p - \hat{\bm \mu}_i)^\top \hat{\bm Q}_i (\bm p - \hat{\bm \mu}_i).
\end{align}

\begin{remark}
In real applications, the fitting function $f_\text{Fit}(\cdot)$ can be implemented using different algorithms, such as the optimization-based methods \cite[Section 8.4]{boyd2004convex} or deep learning-based methods (\cite{10142208}). 
One can choose suitable fitting functions based on the specific application requirements and computational resources.  
\end{remark}

For the case that the measured point clouds are perfect, one can apply the method in \eqref{equ:scaling_factor_obs} to achieve the collision detection between elliptical sets. 
In detail, the minimum scaling factor $\hat\gamma_i^*(\bm x)$ between the robot $\mathcal{R}(\bm x)$ and the $i$-th obstacle $\hat \obs_i$ can be computed through the following optimization problem
\begin{subequations}
\label{equ:scaling_factor}
\begin{align}
\label{equ:scaling_factor_obj}
&\hat\gamma_i^*(\bm x) = \min_{\bm p}~\F_{\hat\obs_i}(\bm p, \bm \theta_i) \\
\label{equ:scaling_factor_cons}
&\text{s.t.}~\F_{\R}(\bm p, \bm \theta_\R(\bm x)) \leq 1. 
\end{align}
\end{subequations}
As seen from the above content, $\hat\gamma_i^*(\bm x)>1$ indicates that the robot and obstacle have not intersected.
However, it should be emphasized that due to measurement noise and limitations in prediction accuracy, it is challenging in real applications to ensure that the constructed obstacle regions $\hat \obs_i$ fully encapsulate all actual point cloud data. 

Motivated by this observation, the case that the point clouds are measured under the sensor noise is considered. 
Following the approach in \cite{yang2023safe}, each observed point is modeled as
\begin{align}
\label{equ:sensor_model}
\bm p_{i,j} =  \bar{\bm p}_{i,j} + \bm \nu_{i,j},
\end{align}
where $\bar{\bm p}_{i,j}$ represents the true position of the point, and $\bm \nu_{i,j}$ is an additive random variable representing the sensor noise, following an unknown distribution. 
For simplicity, let $\bar{P}_i=\{\bar{\bm p}_{i,j}\}_{j=1}^{n_i}$ collect the true values of the $i$-th obstacle's point cloud.
Based on the above discussions, the formal statement of the risk-aware navigation problem under noisy measurement can be stated as follows.

\textbf{Problem Statement:} Consider the system described in (\ref{equ:sys1}) with initial state $\bm x(0) \in \X$ and its ellipsoidal shape $\R(\bm x)$, and the sensor model is given in (\ref{equ:sensor_model}).
Let $T \subset R_{\geq 0}$ be a time interval, and $\alpha$ denote a failure probability.
The objective of this paper is:
\begin{itemize}
\item [1)] Construct another ellipsoid   $\obs_i$ from the robot's sensor measurements $P_i$ for every $i\in[N]$, such that the true point cloud $\bar{P}_i $ is covered by $\obs_i$ in the sense of   
$$
\text{Prob}(\bar{P}_i \subseteq \obs_i) \geq 1 - \alpha.
$$

\item [2)]  Furthermore, design a safety-critical control law $\bm u$ such that the robot could avoid the collision with the surrounding obstacles, i.e., 
$$\R(x(t))\cap(\cap_i^N\obs_i) =\emptyset, ~ \forall t \in T.$$
\end{itemize}

\begin{figure}[t]
  \centering
  \includegraphics[width=0.475\textwidth, trim=145 148 176
   12, clip]{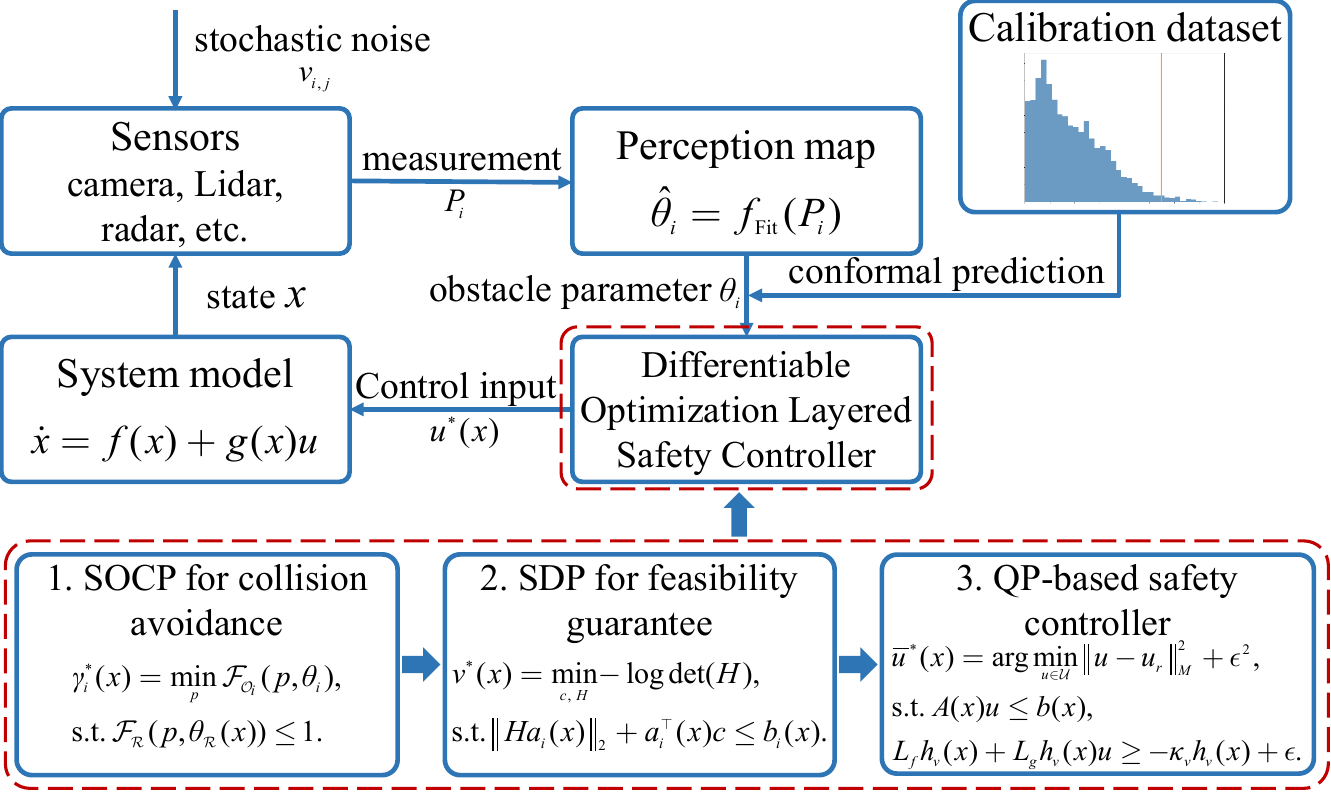}
   \vspace{-0.2cm}
  \caption{Overall framework of the proposed method.}
  \label{fig:framework}
\end{figure}

\section{Main Results}\label{sec:main}
\vspace{-2mm}
In this section, the main results of the paper are presented. The overall framework is illustrated in Fig. \ref{fig:framework}.
First, conformal prediction is applied for calibrating the perception system, enabling the identification of reliable obstacle regions that account for sensor uncertainty. 
Subsequently, a differentiable optimization layered safety controller is proposed to simultaneously guarantee safety and feasibility.
\subsection{Calibrating the Perception System}
Conformal prediction enables rigorous uncertainty quantification in perception, ensuring reliable obstacle detection. 
The first objective of this paper is to establish a robust detection framework that accounts for uncertainty arising from both sensor measurements and prediction errors. 
To achieve this, a calibration set $D_{\text{cal}} = \{E_i\}_{i=1}^{N_{\text{cal}}}$, where $N_{\text{cal}}\in\mathbb{N}$,  is constructed. 
Each element $E_i$ in the calibration set includes the true point cloud $\bar{P}_i$ and the parameters of the predicted obstacle ellipsoid $\hat{\theta}_i$ from the sensor measurements. 
In the calibration step, the goal is to determine the minimal scaling factor $\gamma_i'$ such that the inflated ellipsoid $\hat{\obs}_i'(\gamma_i')$ can fully cover the true point cloud $\bar{P}_i$, i.e. $\bar{P}_i\subseteq \hat{\obs}_i'(\gamma_i')$, where 
\begin{align}
  &\hat{\obs}_i'(\hat{\gamma}_i)=\{\bm p\in\reals^{n_p}|\F_{\hat{\obs}_i}(\bm p, \hat{\bm\theta}_i) \leq \gamma_i'\}.
\end{align}
Here, $\gamma_i'$ is the scaling factor that inflates the original ellipsoid along each axis by a factor of $\sqrt{\gamma_i'}$. 
Similar to \cite{mei2024perceive} and \cite{kwon2025conformalized}, the minimal scaling factor is defined as the nonconformity score for each calibration element, which can be computed by solving the following optimization problem
\begin{subequations}
  \label{equ:scaling_factor_cal}
  \begin{align}
    \label{equ:scaling_factor_cal_obj}
    &\gamma_i' = \max \quad \gamma_{i,j}' \\
    \label{equ:scaling_factor_cal_cons}   
    &\text{s.t.}  \quad \gamma_{i,j}' \geq 1,  \quad \F_{\hat{\obs}_i}(\bar{\bm p}_{i,j},\hat{\bm\theta}_i) \leq \gamma_{i,j}', \quad \forall j\in[n_i].
  \end{align}
\end{subequations}
Referring to \cite{shafer2008tutorial}, given the nonconformity scores over the calibration set $D_{\text{cal}}$ and a user-specified miscoverage level $\alpha$, the nonconformity score $\gamma_i'$ of a new sample $(\bar{P}_i, \hat{\bm \theta}_i)$, satisfies the following probabilistic guarantee
\begin{align}
  \label{equ:cp_guarantee}
  \text{Prob}\left(\gamma_i' \leq \Delta\right) \geq 1 - \alpha,
\end{align}
where $\Delta$ is chosen as the $\lceil (1-\alpha)(N_{\text{cal}} + 1) \rceil$-th quantile of the nonconformity scores from the calibration set.
\begin{thm}
  By inflating the original ellipsoid by a factor of $\sqrt{\Delta}$ along each axis, the true obstacle point cloud is fully contained within the inflated ellipsoid with probability at least $1-\alpha$, i.e.,
  \begin{align}
    \text{Prob}(\bar{P}_i \subseteq \hat{\obs}_i'(\Delta)) \geq 1 - \alpha.
  \end{align}
\label{thm:cp_guarantee} 
\end{thm}
\begin{pf}
Given (\ref{equ:cp_guarantee}), for a new sample $(\bar{P}_i, \hat{\bm \theta}_i)$, it holds that
\begin{align}
  \label{equ:proof_cp}
  &\text{Prob}(\bar{P}_i \subseteq \hat{\obs}_i'(\Delta)) \nonumber\\
  &\geq\text{Prob}(\bar{P}_i \subseteq \hat{\obs}_i'(\gamma_i'))\cdot\text{Prob}(\hat{\obs}_i'(\gamma_i') \subseteq \hat{\obs}_i'(\Delta))\nonumber\\
  &=\text{Prob}(\hat{\obs}_i'(\gamma_i') \subseteq \hat{\obs}_i'(\Delta))
\end{align}
Based on the definition of the scaling function, $\hat{\obs}_i'(\Delta)$ and $\hat{\obs}_i'(\gamma_i')$ are both ellipsoids centered at $\hat{\bm \mu}_i$ and inflated along each axis by a factor of $\sqrt{\Delta}$ and $\sqrt{\gamma_i'}$, respectively. 
Therefore, the condition $\hat{\obs}_i'(\gamma_i') \subseteq \hat{\obs}_i'(\Delta)$ is equivalent to $\gamma_i' \leq \Delta$.
(\ref{equ:proof_cp}) can be rewritten as
\begin{align}
  &\text{Prob}(\bar{P}_i \subseteq \hat{\obs}_i'(\Delta)) \nonumber\\
  &=\text{Prob}\left(\gamma_i' \leq \Delta\right)\geq 1 - \alpha.
\end{align}
Thus, the true obstacle point cloud is fully contained within the inflated ellipsoid with probability at least $1-\alpha$.
\end{pf}
For brevity, the following equations are used to represent the inflated obstacle regions 
\begin{align}
  &\obs_i =\hat{\obs}_i'(\Delta) = \{\bm p\in\reals^{n_p}|\F_{\obs_i}(\bm p, {\bm \theta}_i) \leq 1\},\\
  &\F_{\obs_i}(\bm p, {\bm \theta}_i) = (\bm p - {\bm \mu}_i)^\top {\bm Q}_i (\bm p - {\bm \mu}_i).
\end{align}
where $\bm \mu_i=\hat{\bm \mu}_i$ and ${\bm Q}_i=\hat{\bm Q}_i/\Delta$.

\begin{remark}
Limited calibration set size may cause the actual coverage to differ from the nominal level.
To address this, conformal prediction can offer a dataset-conditional guarantee (\cite{pmlr-v25-vovk12, kwon2025conformalized}) such that, given a fixed calibration set, the probability that the inflated region covers the true obstacle point cloud for all future samples satisfies
\[
  \text{Prob}(\bar{P}_d \subseteq \obs_d|D_{\text{cal}}) \geq \text{Beta}_{N_{\text{cal}} - v + 1, v}(\delta),
\]
where $v = \lfloor \hat{\alpha}(N_{\text{cal}} + 1) \rfloor$ and $\text{Beta}_{N_{\text{cal}} - v + 1, v}(\delta)$ is the $\delta$-quantile of the Beta distribution with parameters $N_{\text{cal}} - v + 1$ and $v$. 
By choosing $\hat{\alpha}$ appropriately, the desired coverage level $1-\alpha$ is achieved with probability at least $1-\delta$ over calibration dataset sampling. 
\end{remark}

\vspace{-0.2cm}
\subsection{Differentiable Optimization Layered Safety Controller}
\vspace{-0.2cm}
{\bf Differentiable optimization-based collision avoidance:} With the inflated obstacle regions, the next step is to design a safety controller that can ensure collision-free navigation. 
To this end, the scaling factor $\gamma_i^*(\bm x)$ between the robot $\R(\bm x)$ and the $i$-th inflated obstacle region $\obs_i$ is constructed as a CBF. In detail, define the scaling factor $\gamma_i^*(\bm x)$ as the optimal value of the SOCP 
\begin{subequations}
  \label{equ:scaling_factor_obs_rob}
  \begin{align}
    &\gamma_i^*(\bm x) = \min_{\bm p}~\F_{\obs_i}(\bm p, \bm \theta_i) \\
    &\text{s.t.}~\F_{\R}(\bm p, \bm \theta_\R(\bm x)) \leq 1. 
  \end{align}
\end{subequations}

According to Lemma 2, $\gamma_i^*(\bm x)$ is continuously differentiable with respect to $\bm x$.
Thus, the CBF for the $i$-th obstacle is defined as
\begin{align}
  h_i(\bm x) = \gamma_i^*(\bm x) - \gamma_0.
\end{align}
where $\gamma_0>1$ is a predefined safety margin. 
And the safe set with respect to the $i$-th obstacle is defined as $\cal{S}_i = \{\bm x\in\mathcal{X}|h_i(\bm x)\geq 0\}$.
By applying the chain rule, the Lie derivatives of $h_i(\bm x)$ can be computed as
\begin{align}
  \label{equ:lie_derivatives}
  L_fh_i(\bm x) &= \frac{\D \gamma_i^*(\bm x)}{\D \bm x} f(\bm x),
  L_gh_i(\bm x) = \frac{\D \gamma_i^*(\bm x)}{\D \bm x} g(\bm x),\\
  \label{equ:chain_rule}
  \frac{\D \gamma_i^*(\bm x)}{\D \bm x} &= \frac{\D \gamma_i^*(\bm x)}{\D \bm \theta_{\R}}\frac{\D \bm \theta_{\R}}{\D \bm x}.
\end{align}
In (\ref{equ:chain_rule}), ${\D \bm \theta_{\R}}/{\D \bm x}$ can be directly calculated based on the definition of $\bm \theta_{\R}(\bm x)$. 
To compute ${\D \gamma_i^*(\bm x)}/{\D \bm \theta_{\R}}$, \cite{10184036} and \cite{11124847} provide rigorous analytical methods. 
However, for cases where the input relative order is one, a simpler approach can be obtained through the following definition and lemma.

\emph{Definition 3 (Regular point):} A feasible point $\bm p^*$ of the optimization problem
\begin{subequations}
  \label{equ:general_optimization}
  \begin{align}
  \label{equ:general_optimization_obj}
  \min_{\bm p} & \quad F_0(\bm p, \bm \theta_{0}) \\
  \label{equ:general_optimization_cons}
  \text{s.t.} & \quad G_i(\bm p, \bm \theta_{i}) \leq 0, \quad i=[m],
\end{align}
\end{subequations}
is called a regular point if the gradients of the active constraints at $\bm p^*$ are linearly independent.  

\emph{Definition 4 (Degenerate inequality constraint):} An inequality constraint $G_i(\bm p, \bm \theta_{i})\leq 0$ is called degenerate at the solution $\bm p^*$ of the optimization problem in (\ref{equ:general_optimization}) if the corresponding Lagrange multiplier is zero.

\emph{Lemma 3 (\cite{castillo2008sensitivity}):} Assume that the solution $\bm p^*$ of the optimization problem (\ref{equ:general_optimization}) is a regular point and that no degenerate inequality constraints exist. 
Then, the sensitivity of the objective function with respect to the parameter $\bm \theta_{0}$ and $\bm \theta_{i}$ is given by the gradient of the Lagrangian function.

To apply this result to problem (\ref{equ:scaling_factor_obs_rob}), it is necessary to verify that the optimal solution satisfies the required regularity conditions and that the inequality constraints is not degenerate. 
The following theorem confirms that these conditions are met.

\begin{thm}\label{th2}   
The optimal solution $\bm p^*$ of the problem (\ref{equ:scaling_factor_obs_rob}) is a regular point and no degenerate inequality constraints exist.
\end{thm}

\begin{pf}    
For (\ref{equ:scaling_factor_obs_rob}), the Lagrangian function can be expressed as
\begin{align}
  \mathcal{L}(\bm p, \lambda_i) = \F_{\obs_i}(\bm p, {\bm \theta}_i) + \lambda_i(\F_{\R}(\bm p, \bm \theta_{\R}) - 1).  
\end{align}
Let $\bm p^*$ and $\lambda_i^*$ denote the optimal solution and the corresponding Lagrange multiplier of (\ref{equ:scaling_factor_obs_rob}), respectively. The KKT conditions yield
\begin{subequations}
  \label{equ:kkt_conditions}
  \begin{align}
    \label{equ:kkt_conditions_grad}
    &\frac{\D \F_{\obs_i}}{\D \bm p}(\bm p^*, {\bm \theta}_i) + \lambda_i^* \frac{\D \F_{\R}}{\D \bm p}(\bm p^*, \bm \theta_{\R}) = 0,\\
    \label{equ:kkt_conditions_primal}
    &\lambda_i^* (\F_{\R}(\bm p^*, \bm \theta_{\R}) - 1) = 0,\\
    \label{equ:kkt_conditions_dual}
    &\lambda_i^* \geq 0.
  \end{align}
\end{subequations} 
For $\F_{\obs_i}$, one have 
\[
  \frac{\D \F_{\obs_i}}{\D \bm p}(\bm p^*, {\bm \theta}_i)=2(\bm p^*-\bm \mu_i)^\top\bm Q_i.
\] 
Due to $\bm p^*\nsubseteq\obs_i$ and $\bm Q_i$ is a positive-definite matrix, it follows that the gradient is non-zero at this point $\bm p^*$. 
Therefore, by (\ref{equ:kkt_conditions_grad}) and (\ref{equ:kkt_conditions_dual}), one have $\lambda_i^* > 0$ and the gradient ${\D \F_{\R}(\bm p^*, \bm \theta_{\R})}/{\D \bm p}$ is also non-zero.
Due to the problem only having one inequality constraint and the gradient of the active constraint is non-zero at $\bm p^*$, it indicates that $\bm p^*$ is a regular point.
Besides, since $\lambda_i^* > 0$, there are no degenerate inequality constraints at the solution $\bm p^*$.
\end{pf}

Based on the result of Theorem \ref{th2}, the sensitivity of the scaling factor $\gamma_i^*(\bm x)$ with respect to the robot parameters $\bm \theta_{\R}$ can be computed through Lemma 3 as
\begin{align}
  \frac{\D \gamma_i^*(\bm x)}{\D \bm \theta_{\R}} = \lambda_i^* \frac{\D \F_{\R}(\bm p^*, \bm \theta_{\R})}{\D \bm \theta_{\R}}, 
\end{align}
where $\lambda_i^*$ can be solved through (\ref{equ:kkt_conditions_grad}). 
Based on Lemma 1, any Lipschitz continuous control input satisfies 
$
\bm u\in\mathcal{U}_{h_i}(\bm x) = \{\bm u\in\reals^m|L_fh_i(\bm x) + L_gh_i(\bm x)\bm u \geq -\kappa_i h_i(\bm x)\}
$ can guarantee the forward invariance of the safe set $\cal{S}_i$.
\begin{remark}
To ensure collision avoidance with all obstacles, the overall safe set is defined as $\cal{S} = \bigcap_{i=1}^N \cal{S}_i$.
The safety controller can be solved through the following QP:
\begin{subequations}
  \label{equ:safety_qp}
  \begin{align}
    \label{equ:safety_qp_obj}
    &\bm u^*(\bm x) = \arg\min_{\bm u\in\mathcal{U}} \|\bm u - \bm u_{r}\|^2_{\bm M} \\
    \label{equ:safety_qp_cons}
    & \text{s.t.}  ~L_fh_i(\bm x) + L_gh_i(\bm x)\bm u \geq -\kappa_i h_i(\bm x), \forall i\in[N],
  \end{align}
\end{subequations}
where $\bm u_{r}$ is a nominal control input, $\bm M$ is a user-defined positive definite matrix.  
\end{remark}

{\bf Differentiable optimization-based feasibility guarantee:} 
While (\ref{equ:safety_qp}) can guarantee the system safety, the feasibility is not addressed explicitly, which may be compromised in the presence of multiple safety and input constraints. 
To address this issue, a volume-CBF is introduced to ensure the persistent feasibility of (\ref{equ:safety_qp}).
Specifically, the feasible space of control input is first defined as $\mathcal{U}_{F}(\bm x) = \{\bm u \in \mathbb{R}^m|\bm A(\bm x)\bm u\leq \bm b(\bm x)\}$, where
\[
{\bm{A}}(\bm x)=
\begin{bmatrix}
\bm A_m \\
-L_{g}h_1(x) \\
\cdots \\
-L_{g}h_N(x)
\end{bmatrix},~
{\bm{b}}(\bm x)=
\begin{bmatrix}
\bm b_m \\
\kappa_1 h_1(\bm x)+L_{f}h_1(x) \\
\cdots \\
\kappa_N h_N(\bm x)+L_{f}h_N(x)
\end{bmatrix}.
\]
For the optimization problem (\ref{equ:safety_qp}), its feasibility can be characterized by the volume of the largest inscribed ellipsoid within the feasible space $\mathcal{U}_{F}(\bm x)$. 
In other words, (\ref{equ:safety_qp}) admits a feasible solution as long as the ellipsoid maintains a strictly positive volume. 
To quantitatively evaluate the feasibility, the inscribed ellipsoid of the feasible space $\mathcal{U}_F(\bm x)$ can be obtained by solving the following SDP
\begin{subequations}
  \label{equ:inscribed_ellipsoid}
  \begin{align}
    \label{equ:inscribed_ellipsoid_obj}
    &v^*(\bm x)=\min_{\bm c, \bm H}  -\log\det(\bm H) \\
    \label{equ:inscribed_ellipsoid_cons}
    &\text{s.t.}~\|\bm H \bm a_i(\bm x)\|_2+\bm a_i^\top(\bm x)\bm c\leq b_i(\bm x),~\forall i\in[2m+N],
  \end{align}
\end{subequations}
where $\bm c\in\reals^m$ and $\bm H\in\reals^{m\times m}$ are the center and shape matrix of the inscribed ellipsoid.
$\bm a_i^\top(\bm x)$ is the $i$-th row of $\bm A(\bm x)$ and $b_i(\bm x)$ is the $i$-th element of $\bm b(\bm x)$. 
Let $\bm c^*$ and $\bm H^*$ be the optimal solutions. The volume of the largest inscribed ellipsoid is given by 
\[
V(\bm x) = \frac{\pi^{\frac{m}{2}}}{\Gamma(\frac{m}{2}+1)}\det(\bm H^*(\bm x))=\frac{\pi^{\frac{m}{2}}}{\Gamma(\frac{m}{2}+1)}e^{-v^*(\bm x)}.
\]
As described above, the feasibility of (\ref{equ:safety_qp}) can be guaranteed by ensuring the forward invariance of the set $\mathcal{S}_F = \{\bm x\in\mathcal{X}|V(\bm x)>0\}$. 
To tighten the feasibility requirement, a user-defined positive constant $V_0>0$ is introduced.
According, the volume-CBF can be defined as
\begin{align}
  h_v(\bm x) = V(\bm x) - V_0.
\end{align}
Similar to (\ref{equ:lie_derivatives}), the Lie derivatives of $h_v(\bm x)$ can be computed through the chain rule as
\begin{align}
  &L_fh_v(\bm x) = \frac{\D V(\bm x)}{\D \bm x} f(\bm x),~
  L_gh_v(\bm x) = \frac{\D V(\bm x)}{\D \bm x} g(\bm x),
\end{align}
\begin{align}
  \label{equ:chain_rule_volume}
  &\frac{\D V(\bm x)}{\D \bm x} = \frac{\D V(\bm x)}{\D v^*}\frac{\D v^*(\bm x)}{\D \bm x}=\frac{-\pi^{\frac{m}{2}}e^{-v^*(\bm x)}}{\Gamma(\frac{m}{2}+1)}\frac{\D v^*(\bm x)}{\D \bm x},\\
  \label{equ:chain_rule_volume1}
  &\frac{\D v^*(\bm x)}{\D \bm x} =\sum_{j\in \mathcal{A}}  \left(\frac{\D v^*(\bm x)}{\D \bm a_j}\frac{\D \bm a_j^\top(\bm x)}{\D \bm x}\!+\!\frac{\D v^*(\bm x)}{\D b_j}\frac{\D b_j(\bm x)}{\D \bm x}\right),
\end{align}
where $\mathcal{A}=\{j\in[2m+N]|\|\bm H^*\bm a_j(\bm x)\|_2+\bm a_j^\top(\bm x)\bm c^*=b_j(\bm x)\}$ is the index set of the active constraints at the solution.
As described in Lemma 2, $\gamma_i^*$ obtained through (\ref{equ:scaling_factor_obs_rob}) is $C_k, k\geq1$ in $\bm \theta$.
Therefore, the derivatives ${\D \bm a_j^\top(\bm x)}/{\D \bm x}$ and ${\D b_j(\bm x)}/{\D \bm x}$ can be directly calculated based on the definitions of $\bm A(\bm x)$ and $\bm b(\bm x)$.
To compute ${\D v^*(\bm x)}/{\D \bm a_j}$ and ${\D v^*(\bm x)}/{\D b_j}$, Lemma 3 can again be utilized.
Given the Lagrangian function of (\ref{equ:inscribed_ellipsoid}) in optimality conditions
\begin{align}
  \mathcal{L}(\bm c^*, \bm H^*, \bm \lambda^*_E) = &-\log\det(\bm H^*) + \sum_{j\in \mathcal{A}} \lambda_{E,j}^*\Big(\|\bm H^*\bm a_j(\bm x)\|_2
  \nonumber\\ & +\bm a_j^\top(\bm x)\bm c^* - b_j(\bm x)\Big),
\end{align}
where $\bm \lambda^*_E$ is the vector of Lagrange multipliers.
The sensitivity of the optimal value $v^*(\bm x)$ with respect to the parameters $\bm a_j$ and $b_j$ can be computed through the gradient of the Lagrangian function as
\begin{align}
  \!\!\frac{\D v^*(\bm x)}{\D \bm a_i} \!\!= \!\lambda_{E,j}^* \!\left(\!\frac{\bm H^{*2}\bm a_j(\bm x)}{\|\bm H^*\bm a_j(\bm x)\|_2} \!+\! \bm c^*\!\right)\!^\top\!,
  \frac{\D v^*(\bm x)}{\D b_j} \!\!=\! -\lambda_{E,j}^*.
\end{align}

Based on Lemma 1, any Lipschitz continuous control input satisfying
$\bm u\in\mathcal{U}_{h_v}(\bm x) = \{\bm u\in\reals^m|L_fh_v(\bm x) + L_gh_v(\bm x)\bm u \geq -\kappa_v h_v(\bm x)\}$
can guarantee the forward invariance of the feasible set $\cal{S}_F$.

Consequently, to enforce the compatibility of CBF and input constraints and ensure that the system always remains within safe operational limits, the controller can be obtained by solving the following QP
\begin{subequations}
  \label{equ:safety_feasibility_qp}
  \begin{align}
    \label{equ:safety_feasibility_qp_obj}
    &[\bm u^*(\bm x), \epsilon^*(\bm x)] = \arg\min_{\bm u\in\mathcal{U}, \epsilon\geq 0} \|\bm u - \bm u_{r}\|^2_{\bm M} + \epsilon^2\\
    \label{equ:safety_feasibility_qp_cons}
    & \text{s.t.}  ~\bm A(\bm x)\bm u\leq\bm b(\bm x),\\
    \label{equ:safety_feasibility_qp_cons2}
    &L_fh_v(\bm x) + L_gh_v(\bm x)\bm u \geq -\kappa_v h_v(\bm x)-\epsilon.
  \end{align}
\end{subequations}
\begin{thm}   
  Suppose that the optimal solution of (\ref{equ:safety_feasibility_qp}) $\bm u^*(\bm x)$ is locally Lipschitz continuous. 
  Under this control law, if the initial state $\bm x(0)\in\mathcal{S}\cap\mathcal{S}_F$ and $\epsilon^*(\bm x)\leq \kappa_v V_0$, then the system state $\bm x(t)$ will remain in $\mathcal{S}\cap\mathcal{S}_F$ for all $t\geq 0$, perserving the safety and feasibility of the system.
\end{thm}
\begin{figure*}[htbp]
    \centering
    \begin{subfigure}[t]{0.325\linewidth}
        \centering
        \includegraphics[width=\linewidth]{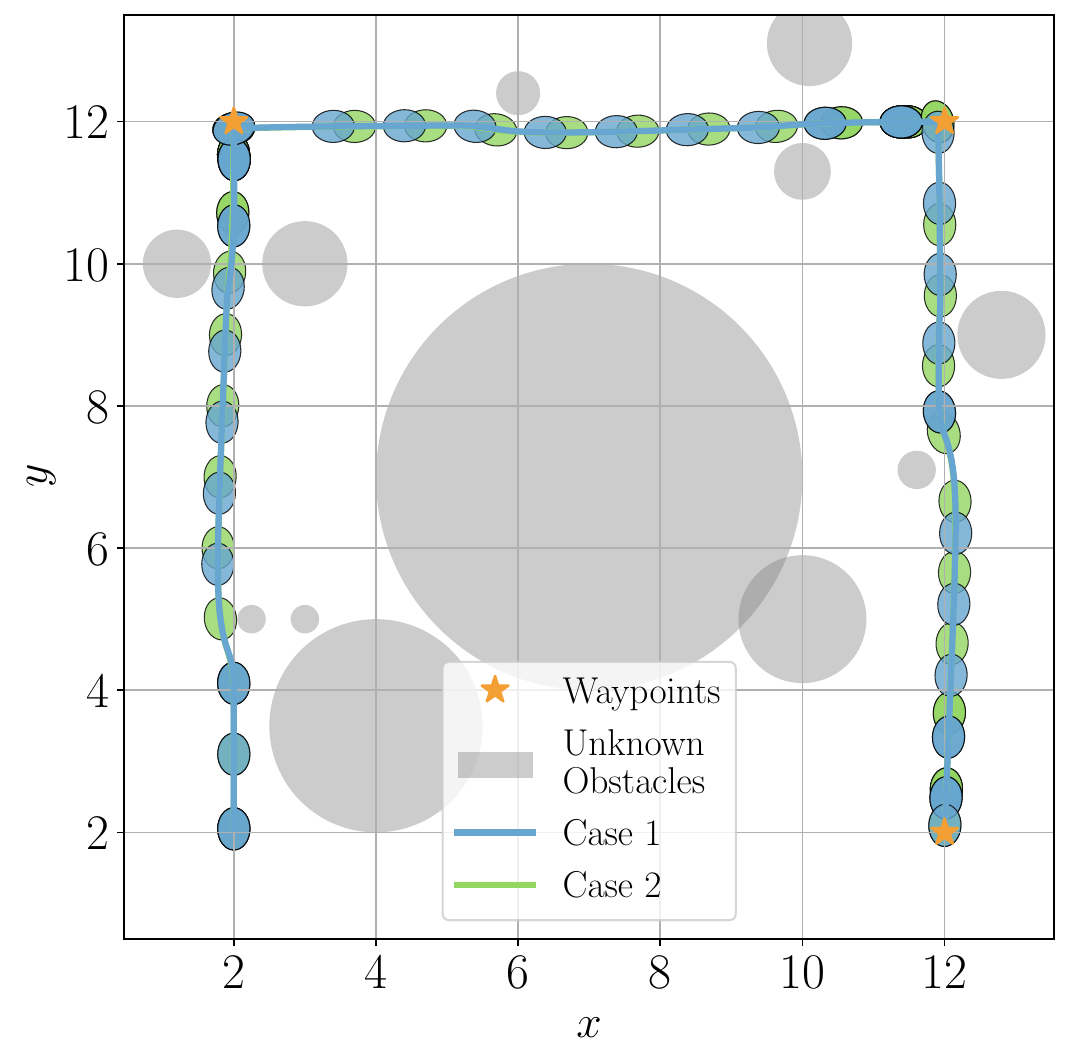}
        \caption{}
        \label{fig:sub_a}
    \end{subfigure}
    \hfill
    \begin{subfigure}[t]{0.325\linewidth} 
        \centering
        \includegraphics[width=\linewidth]{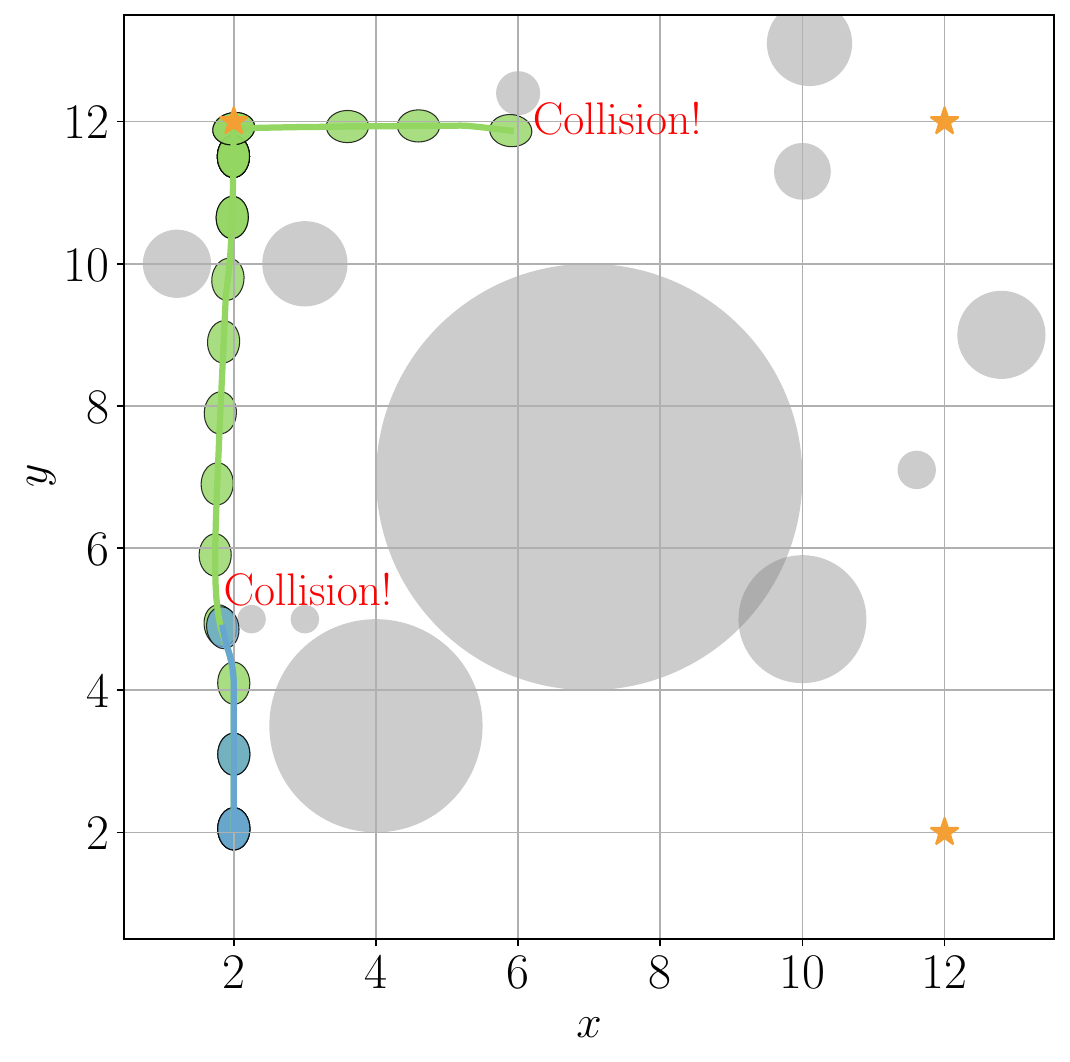} 
        \caption{}
        \label{fig:sub_b}
    \end{subfigure}
    \hfill
    \begin{subfigure}[t]{0.325\linewidth} 
        \centering
        \includegraphics[width=\linewidth]{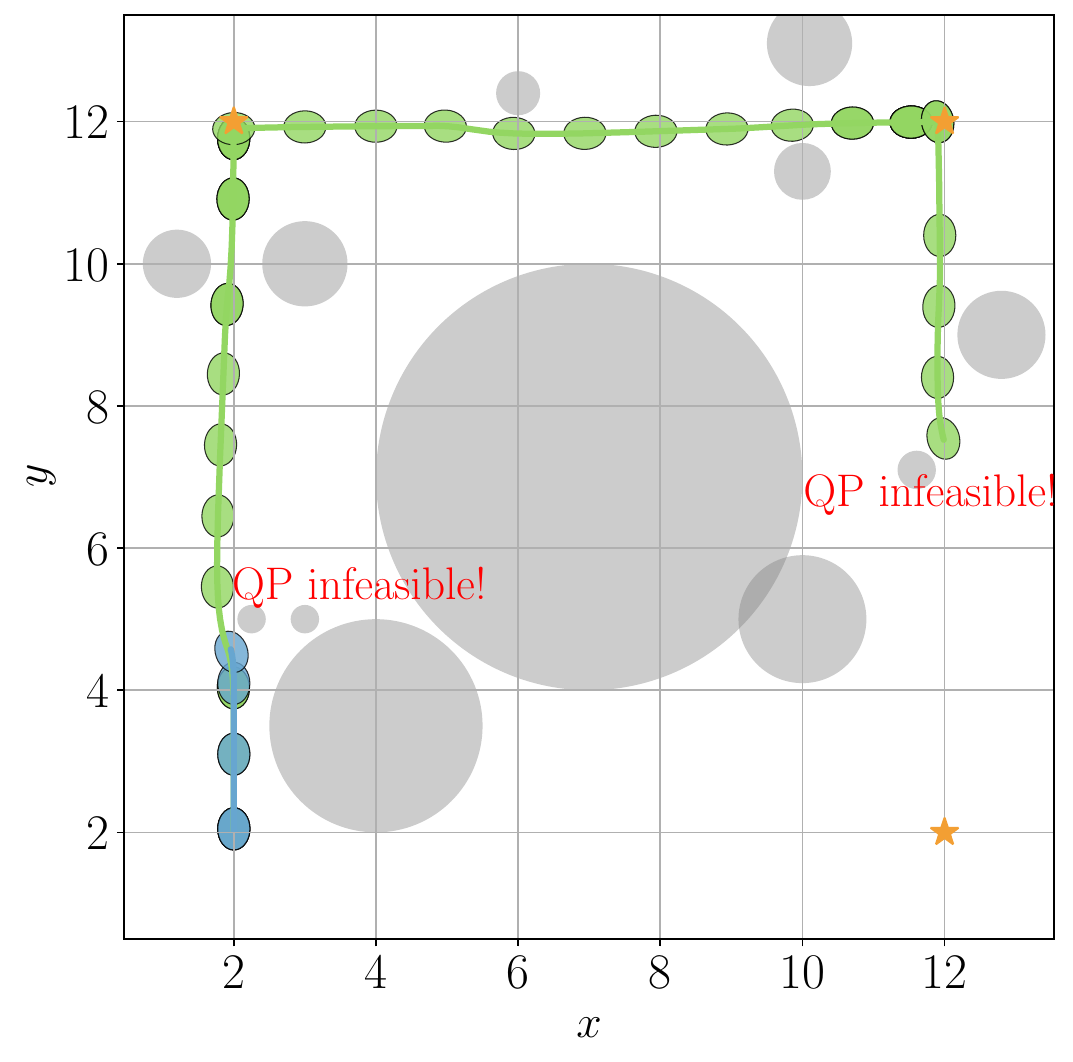} 
        \caption{}
        \label{fig:sub_c}
    \end{subfigure}
    \vspace{-2mm}
    \caption{Paths of the mobile robot under different methods: (a) results for the proposed method; (b) results for Compared method 1; (c) results for Compared method 2.}
    \vspace{-0.2cm}
    \label{fig:sim_trajectory}
\end{figure*}
\begin{pf}
  Due to the assumption that $\bm u^*(\bm x)$ is locally Lipschitz continuous, the existence and uniqueness of the system solution $\bm x(t)$ is guaranteed. 
  When $\epsilon^*(\bm x)\leq \kappa_v V_0$, the constraint (\ref{equ:safety_feasibility_qp_cons2}) can be reformulated as
    \begin{align}
    &L_fh_v(\bm x) + L_gh_v(\bm x)\bm u 
    = L_fV(\bm x) + L_gV(\bm x)\bm u \nonumber\\
    \geq& -\kappa_v (V(\bm x)-V_0)-\epsilon^*(\bm x)\nonumber\\\geq& -\kappa_v V(\bm x).
  \end{align}
  According to Lemma 1, the forward invariance of the feasible set $\cal{S}_F$ is guaranteed, i.e. $V(\bm x(t))> 0$, $\forall t\geq 0$. This condition implies that the CBF and the input constraints are compatible and the feasible input space $\mathcal{U}_F(\bm x)$ is non-empty for all $t\geq 0$. 
  Moreover, the constraint in (\ref{equ:safety_feasibility_qp_cons}) is equivalent to the constraint conditions specified in the formula (\ref{equ:safety_qp}). This equivalence guarantees that the forward invariance of the safe set $\cal{S}$ is preserved, i.e., $\bm x(t)\in\cal{S}$, $\forall t\geq 0$. 
  By combining these results, $u^*(\bm x)$ defined in (\ref{equ:safety_feasibility_qp}) can render the system state $\bm x(t)\in\mathcal{S}\cap\mathcal{S}_F$, $\forall t\geq 0$ if the initial state $\bm x(0)\in\mathcal{S}\cap\mathcal{S}_F$ and $\epsilon^*(\bm x)\leq \kappa_v V_0$.
\end{pf}

\vspace{-0.2cm}
\section{Simulation Results}\label{sec:sim}
\vspace{-0.2cm}

\begin{figure}[t]
    \centering
    \begin{subfigure}[t]{0.95\linewidth} 
        \centering
        \includegraphics[width=\linewidth]{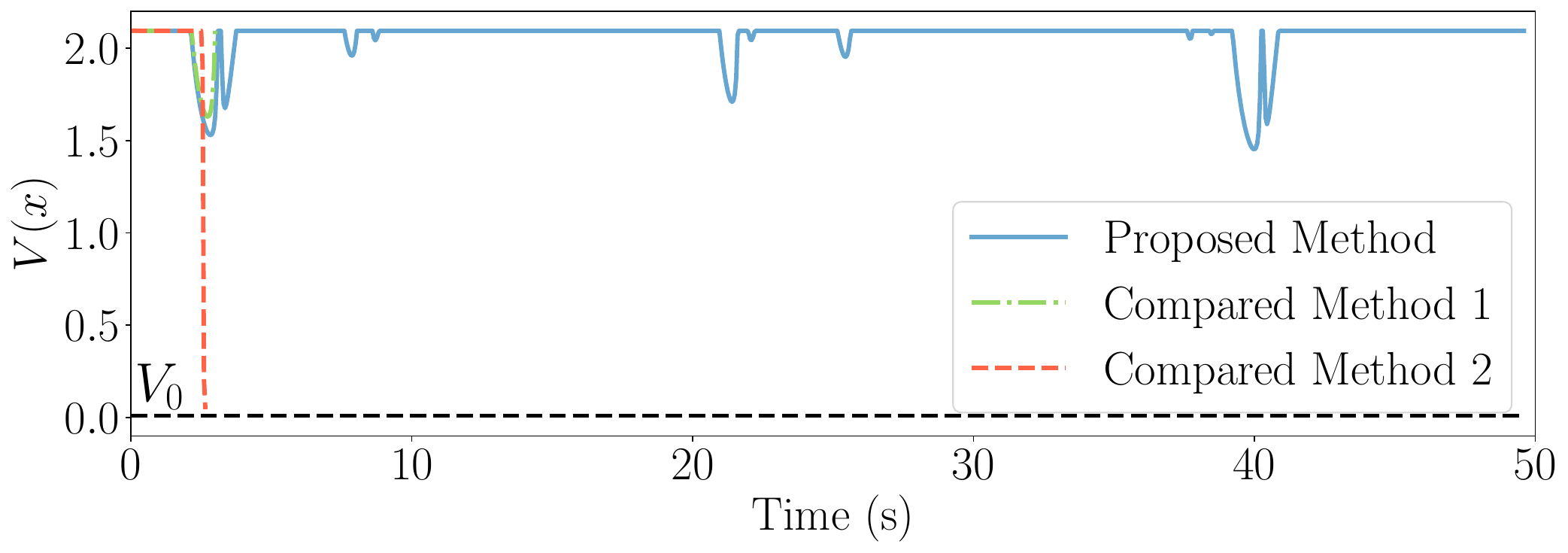} 
        \caption{}
        \label{fig:sub_a}
    \end{subfigure}
    \begin{subfigure}[t]{0.95\linewidth} 
        \centering
        \includegraphics[width=\linewidth]{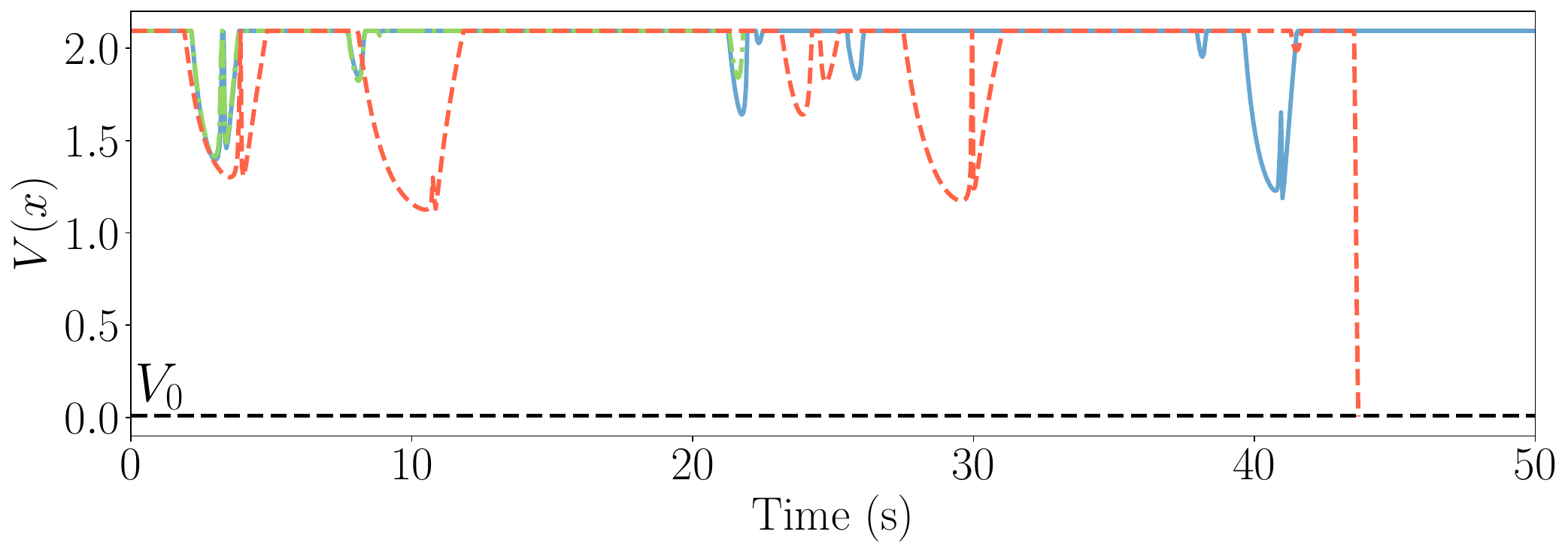} 
        \caption{}
        \label{fig:sub_d}
    \end{subfigure}
    \vspace{-2mm}
    \caption{$V_v(\bm x)$ of different method: (a) results for case 1 ; (b) results for case 2.}
    \label{fig:volume}
        \vspace{-2mm}
\end{figure}

In this section, a case study of a mobile robot navigating through an environment populated with randomly placed obstacles is presented to validate the effectiveness of the proposed method. Referring to \cite{10930526}, the mobile robot is modeled as 
\begin{align*}
\dot{p}_x = v\cos\theta, \dot{p}_y = v\sin\theta, \dot{\theta} = \omega,
\end{align*}
where $\bm x = [p_x,p_y,\theta]^\top\in\reals^3$ represents the state of the robot, control input $\bm u=[v,\omega]^\top$ is bounded within the box constraint $\mathcal{U}=[-1,1]\times[-0.5, 0.5]$. 
In the simulations, the robot is equipped with a depth sensor that provides point cloud measurements of the surrounding obstacles. 
Following \cite{yang2023safe}, unknown noise conforming to an exponential distribution is added to each ray, that is $\bm \nu_{i,j}\sim \exp(\eta)$ with $\eta=2/3$. 
The nominal control input $\bm u_{r}$ is designed to drive the robot towards the predefined goal position $\bm p_d=[p_{dx}, p_{dy}]^\top$
\begin{align*}
&\theta_{r} = \text{atan}(\frac{{p_{dy}}-p_y}{{p_{dx}}-p_x}), ~e_{\theta}=\theta - \theta_{r}, ~\omega =-k_{\omega} e_{\theta}\\
&v_{r} = k_p\sqrt{(p_x-p_{dx})^2 + (p_y-{p_{dy}})^2}\cos(e_{\theta}),
\end{align*}
where $k_{\omega}, k_p$ are positive constants.
The numerical simulations are conducted in Python on a laptop with an Intel Core i9-13900HX CPU and 16GB RAM. 
The optimization problems in (\ref{equ:scaling_factor_obs_rob}), (\ref{equ:inscribed_ellipsoid}), and (\ref{equ:safety_feasibility_qp}) are solved using CVXPY with the SCS solver. 
The parameters are set as $\kappa_i=3.3, \kappa_v=1.1, \gamma_0 = 1.2, V_0=0.01$.  

A calibration dataset of size $N_{\text{cal}}=5000$ is constructed to calibrate the perception system using conformal prediction, ensuring a coverage level of $1-\alpha=0.95$. Thus, the scaling factor of the obstacles is computed as $\Delta=1.346$.
To evaluate the performance of the proposed method, comparisons are made with two baseline approaches. The first baseline employs uninflated obstacle regions (Compared method 1), representing a conventional perception strategy without risk calibration. The second baseline utilizes a standard CBF-based controller without feasibility guarantees (Compared method 2), reflecting the prevailing safety-critical control strategy in the existing literature (\cite{10184036,11124847}).

\begin{table}[t]
\caption{Simulation Results for multiple trials.}
\label{tab:sim_results}
\begin{center}
\begin{tabular}{c c c}
\hline
Method & Success Rate  & Runtime (s) \\
\hline
Proposed Method &  \bf 0.961 & \bf 29.93 \\
Compared Method 1 &  0.907 & 29.35 \\
Compared Method 2 &  0.946 & 26.38 \\
\hline
\end{tabular}
\end{center}
\end{table}

Before each trial, a random seed for noise generation is initialized. Fig. \ref{fig:sim_trajectory} presents two representative cases, showing the paths of the mobile robot under different methods in an environment with multiple obstacles. 
It can be seen that the proposed method successfully guides the robot to reach the goal position while avoiding collisions with obstacles, even in the presence of sensor noise.
Compared method 1 fails to prevent collisions. It can be seen from Fig. \ref{fig:volume} that when the robot is stopped, the volume of the inscribed ellipsoid does not decrease to zero, indicating that the QP remains feasible. But due to the lack of proper calibration, the robot collides with the obstacles.
Compared method 2 encounters infeasibility issues, leading to the robot being unable to reach the goal position.

To avoid the randomness of simulation results, multiple trials are conducted for each method with different random seeds.
The success rate and the average simulation runtime are used as the performance metrics. 
The results are summarized in Table \ref{tab:sim_results}.
Among all evaluated methods, the proposed method achieves the highest success rate of 96.1\%, demonstrating its effectiveness in ensuring safe navigation despite sensor uncertainties.
Besides, the average runtime of the proposed method is comparable to the other methods, indicating the volume-CBF can ensure the compatibility of multiple constraints.

\section{Conclusion and Future Work}\label{sec:conclusion}
This paper presents a novel approach for robotic risk-aware navigation in unknown environments with online perception. 
The proposed framework predicts minimum bounding ellipsoids in real time and employs conformal prediction to address sensor uncertainty. 
By integrating collision and feasibility constraints as control barrier functions via two differentiable optimization layers, the safety-critical controller enables robust obstacle avoidance. 
Numerical simulations validate the effectiveness of the approach for safe navigation in unknown environments. 
Future work will further explore dynamic environments and conduct real-world experiments to comprehensively assess the effectiveness and robustness of the proposed method.


\bibliography{ifacconf}             

\appendix
\end{document}